\documentclass[acmsmall,nonacm]{acmart}

\setcopyright{none}
\copyrightyear{2021}
\acmYear{2021}
\acmDOI{}
\acmConference[]{}{}{}
\acmBooktitle{}
\acmPrice{}
\acmISBN{}
\usepackage{xspace,graphicx,relsize,bm,mathtools,xcolor,xspace}
\usepackage{mathrsfs}
\usepackage{lipsum}

\usepackage{multirow}
\usepackage[capitalise]{cleveref}
\usepackage{thm-restate}
\usepackage{array}
\usepackage{soul}

\usepackage[usestackEOL]{stackengine}

\usepackage{enumitem} %

\usepackage{tabularx}

\newcommand{\eps}{\varepsilon}

\DeclareMathOperator{\col}{col}

\DeclareMathOperator{\polylog}{polylog}

\newtheorem{claim}{Claim}
\newtheorem{remark}{Remark}

\newtheorem{protocol}{Protocol}
\definecolor{protocol-bg}{RGB}{240,240,240}
\usepackage[framemethod=tikz]{mdframed}
\surroundwithmdframed[outerlinewidth=0pt,
innerlinewidth=0pt,
middlelinewidth=0pt,
backgroundcolor=protocol-bg
]
{protocol}

\def\01{\{0,1\}}

\DeclareMathOperator{\negl}{negl}

\usepackage{array}
\newcolumntype{M}[1]{>{\centering\arraybackslash}m{#1}}

\def\ID{\mathrm{ID}}
\def\mL{\mathcal{L}}

\renewcommand\thmcontinues[1]{Formal Statement}

 \definecolor{greenn}{rgb}{0.2,0.6,0.2}
\definecolor{bluue}{rgb}{0.7,0,0.3}
\newcommand{\bbF}{\mathop{\mathbb{F}}}
\newcommand{\fq}{{\textstyle \bbF_q}}
\newcommand{\tfree}{\Delta_\mathsf{free}}
\DeclareMathOperator{\poly}{poly}
\newcommand{\dzma}{\mbox{\sf dNIZK}\xspace}
\newcommand{\dma}{\mbox{\sf dNI}\xspace}
\newcommand{\con}{\mathcal{G}}

\newcommand{\cgst}{\textsc{congest}\xspace}

\def\submission{1}

\ifx\submission\undefined

\newcommand{\remove}[1]{}

\else
\newcommand{\remove}[1]{}

\fi

\newcommand{\myparagraph}[1]{\smallskip \noindent {\em #1}.}

\begin{document}

\title{Distributed Non-Interactive Zero-Knowledge Proofs}

\author{Alex B. Grilo}
\email{alex.bredariol-grilo@lip6.fr}
\affiliation{%
  \institution{Sorbonne Université, CNRS, LIP6}
  \country{France}
}
\author{Ami Paz}
\email{}
\affiliation{%
  \institution{Université Paris-Saclay, CNRS, LISN}
  \country{France}
}
\author{Mor Perry}
\email{morpy@mta.ac.il}
\affiliation{%
  \institution{The Academic College of Tel-Aviv-Yaffo}
  \country{Israel}
}

\begin{abstract}
Distributed certification is a set of mechanisms that allows an all-knowing prover to convince the units of a communication network that the network's state has some desired property, such as being $3$-colorable or triangle-free. Classical mechanisms, such as \emph{proof labeling schemes} (PLS), consist of a message from the prover to each unit, followed by one round of communication between each unit and its neighbors. 
Later works consider extensions, called \emph{distributed interactive proofs}, where the prover and the units can have multiple rounds of communication before the communication among the units.
Recently, Bick, Kol, and Oshman (SODA '22) defined a zero-knowledge version of distributed interactive proofs, where the prover convinces the units of the network’s state without revealing any other information about the network’s state or structure. 
In their work, they propose different variants of this model and show that many graph properties of interest can be certified with them.

In this work, we define and study \emph{distributed non-interactive zero-knowledge proofs} (dNIZK);
these can be seen as a non-interactive version of the aforementioned model, and also as a zero-knowledge version of PLS. 
We prove the following: 

- There exists a dNIZK protocol for $3$-coloring with  $O(\log n)$-bit messages from the prover and $O(\log n)$-size messages among neighbors. This disproves a conjecture from previous work asserting that the total number of bits from the prover should grow linearly with the number of edges.  

- There exists a family of dNIZK protocols for triangle-freeness, that presents a trade-off between the size of the messages from the prover and the size of the messages among neighbors. Interestingly, we also introduce a variant of this protocol where the message size depends only on the maximum degree of a node and not on the total number of nodes, improving upon the previous non-zero-knowledge protocol for this problem.  

-  There exists a dNIZK protocol for any graph property in NP in the random oracle models, which is secure against an arbitrary number of malicious parties. Previous work considered compilers from PLS to distributed zero-knowledge protocol, which results in protocols with parameters that are incomparable to ours.
\end{abstract}

\maketitle

\clearpage
\tableofcontents

\clearpage

\section{Introduction}
Interactive (and non-interactive) proof systems are a cornerstone of complexity theory, having important applications from verifiable delegation of computation to inapproximability results. An interactive proof system consists of a two-party protocol, where an unbounded prover interacts with a polynomially-bounded verifier to convince the latter that a given statement is true. 
Such a system is characterized by two desired properties:
completeness, which states that if the statement is indeed true then the prover can make the verifier accept; 
and soundness, which requires that if the statement is false then the prover cannot make the verifier accept, except with small probability.

Given their importance, interactive proof systems have been extended to many settings, and in this work we are interested in two of these extensions. 
In a cryptographic context, we are also interested in a third property of interactive proof systems: zero-knowledge (ZK). 
This property asks that if the statement is indeed true, then the verifier does not learn anything from the interaction with the prover.
This is formalized (in the centralized setting) by asking for a randomized polynomial-time algorithm, called a {\em simulator}, whose output distribution is the same as the distribution of the transcript of the original protocol.
ZK protocols are fundamental in modern cryptography, being an important building block of several practically relevant constructions. 
For its practical applications, one of the most important family of zero-knowledge proofs are non-interactive zero-knowledge proofs (NIZK), where the prover sends a single message to the verifier.~\footnote{We notice that in NIZKs, there is also a trusted setup, such as the Random Oracle Model, common random/reference strings, and others. In the centralized setting, NIZKs are impossible for hard languages without such a trusted setup, i.e. in the {\em plain model}.}

In parallel, interactive proof systems have also been considered in the distributed setting, where a know-all prover wants to convince the nodes of a network of some statement regarding the network's state. As before, the properties that we seek in these distributed protocols are completeness and soundness. 
One of the main models in this setting are {\em proof labeling schemes} (PLS), where the prover sends a single message to each node, followed by one round of communication between each node and its neighbors. 
This notion was also extended to more general distributed interactive proofs where there is an interaction between the prover and the nodes.

Motivated by the abundance of applications of zero-knowledge proofs in intrinsically distributed settings (see, e.g.,~\cite{Ben-SassonCG0MTV14}), Bick, Kol, and Oshman~\cite{BickKO22} recently combined these two worlds together and proposed the notion of {\em
distributed zero-knowledge proofs} (dZK).
Their first contribution was to properly
{\em define} these objects, proposing two suitable definitions. In the first, the goal is to prevent the nodes from learning information from the prover's global view of the graph. 
In the second, which they call {\em strong} distributed zero-knowledge, the goal is to prevent the nodes to learn {\em
anything} from the protocol, including information that could be provided by
their neighbors. 
They then present distributed zero-knowledge protocols for several problems of interest such as $3$-coloring, spanning tree verification, and distributed problems with efficient PLS.

Since the number of rounds is an undesired bottleneck in communication protocols, our goal in this work is to achieve distributed zero-knowledge with optimal round complexity. 
Concretely, we study the non-interactive version of distributed zero-knowledge proofs, which we call $\dzma$.
In this version, the prover sends only one message to each node, and the nodes do not answer back. The nodes use one communication round in the network, and output their result. 
In addition, our notion of distributed zero-knowledge is the {\em strong} distributed zero-knowledge from previous work.
First, 
we show a $\dzma$ protocol for $3$-coloring with a total of $O(n \log n)$ bits of communication with the prover, disproving a conjecture from~\cite{BickKO22} on the required amount of communication for distributed ZK protocols for the problem. 
Second, we propose a $\dzma$  protocol for triangle-freeness, which presents an interesting trade-off between the size of messages from the prover and  and size of messages between neighbors in the graph.
Finally, we provide a universal $\dzma$ for any graph property in NP.

\subsection{Our Model}

We now survey our model; the formal definitions can be found in \cref{sec:prelim}.
We consider a protocol between a prover and a network of nodes, where the prover has full knowledge of the network while the nodes only know their neighborhood. 
The prover wants to convince the nodes that the network's communication graph has a desired property, using a protocol with the following structure.
\begin{enumerate}
    \item The prover sends a message $\sigma_v$ to each node $v$;
    \item Based on $\sigma_v$ and a random string $r$ shared by the nodes, $v$ sends a message $\gamma_{vu}$ to each neighbor $u$;
    \item Based on $r$, $\sigma_v$ and all the received messages $\gamma_{uv}$, node $v$ decides whether to accept or reject.
\end{enumerate}

We say such a protocol is a distributed non-interactive zero knowledge $\dzma(\ell, \mu)$
if $|\sigma_v|\in O(\ell)$ and $|\gamma_{vu}| \in O(\mu)$ for all nodes $u,v$, and in addition it has the following properties.

First,
we require the two standard properties of proof systems: completeness, which states that if the network graph indeed has the desired property, the prover can convince all nodes about it; and soundness, which says that if the network does not have the desired property, at least one node will reject with high probability, for all possible strategies of the prover.
And second, the zero-knowledge property: all the messages received by any node in the protocol can be {\em simulated}. 
More precisely, we want a randomized algorithm that without communicating with the neighbors, 
outputs at node $v$ the values $(\tilde{r},\tilde{\sigma}, (\tilde{\gamma}_{uv})_u)$ following the same distribution of the messages received by $v$ in the real protocol.

As discussed in more detail in~\cite{BickKO22}, some well-knowledge PLS/distributed interactive proofs are trivially zero-knowledge such as $2$-coloring: if the graph is $2$-colorable, when a node receive its color, it knows the color of all of its neighbors (and we can turn this intuition into a simulator). However, achieving zero-knowledge for other problems is much harder. For instance, the standard PLS for $3$-coloring is not zero-knowledge:  each node learns the colors of all of its neighbors and this leaks their correlation in some cases.

We focus on zero-knowledge against a \emph{single, malicious} node. 
Our definition can be extended to multiple nodes, and we state it explicitly when this stronger security holds.

\subsection{Results and Techniques}
\label{sec:our-results}

\subsubsection{Coloring and Polynomial Sharing}

We start in \cref{sec:threecol-polynom}, where we study distributed ZK for $3$-colorability. 
This was a main case study in the previous work~\cite{BickKO22}, 
where they provide a protocol with $3$ communication rounds  with the prover ({\sf MAM})~\footnote{As they mention, their protocol can be compressed to a single round of communication if the prover and the nodes share correlated randomness, i.e. a pair of values that are sampled from a joint distribution and shared among the prover and nodes. This is a rather strong assumption, and it allows, for example, to achieve some primitives that are impossible even in the random oracle model~\cite{Passs05,IshaiKMOP13}.}.
Each node in their protocol communicates $\Theta(\Delta)$ bits with the prover, for a total of $\Theta(m)$ bits in the worst case.
They then present an open question: can one show that any distributed zero-knowledge protocol for
$3$-coloring requires 
$\Omega(m)$ communication with the prover? 

Intuitively, verifying a proper coloring requires checking inequality on each edge, and doing so in zero-knowledge should require at least \emph{one bit of help per edge}, hence the suggested bound.
This is done in~\cite{BickKO22} using a beautiful protocol, inspired by a classical ZK protocol for graph non-isomorphism, which requires two more interaction rounds after the color assignment.
We present a different protocol, where along with colors, each node gets extra $O(\log n)$ bits that allow it to verify the color-inequality with its neighbors.
Doing so, we answer the open question on the negative, as our $\dzma$ protocol only uses $O(n\log n)$ bits of the proof in total.
Moreover, our protocol is safe against a malicious adversary (which might deviate from the protocol) and not only against an honest-but-curious one as the previous protocol.
This is done at the mild expense of $\O(\log n)$ bits exchanged  between neighbors, compared to $O(1)$ in \cite{BickKO22}.
We prove the following.

\begin{restatable*}{thm}{thmcoloring}
	\label{thm:3col}
	$3$-col $\in \dzma\left(\log n,\; \log n\right)$.
\end{restatable*}

The proof of this theorem introduces a new technique, which we call \emph{polynomial sharing}, due to its resemblance to secret sharing.

The proof starts similarly to classical zero-knowledge proofs of $3$-colorability, where Merlin assigns each node $u$ a color $\col(u)$ derived from a random permutation applied to an initial $3$-coloring of the graph.
Each node $u$ can then locally create a low-degree polynomial $C_u$, which is an indicator polynomial of its color, i.e., $C_u(i)=1$ iff $\col(u)=i$.

From these polynomials, we can define a polynomial $P_u$ for each node $u$, as
$P_u=\sum_{v\in N(u)} C_u C_v$, with the desired property that $P_u(i)=0$ for $i=0,1,2$ iff the color of $u$ is different from the colors of all its neighbors. 
Hence, by examining $P_u$ at the points $0,1,2$, node $u$ can verify that its color is unique among itself and its neighbors. 
However, $P_u$ contains significant information about the colors of the neighbors, so it cannot be revealed to node $u$ without violating the zero-knowledge property.
This calls for using the polynomial sharing technique, as described next.

Instead of revealing $P_u$ to $u$,  $P_u$ is shared among $u$ and all its neighbors in a way that leaks no information to either of them, yet still allows them to jointly verify the required properties of $P_u$.
To this end, each neighbor $v$ of $u$ gets a \emph{helper polynomial} $H_v$, and $u$ gets a polynomial $P^{(0)}_v$ such that 
$P_u=P^{(0)}_u+\sum_{v\in N(u)}H_v$.
In this way, $u$ can evaluate $P_u$ in point $i$ by collecting the evaluations $H_v(i)$ from all its neighbors and summing it with its evaluation of $P^{(0)}_u(i)$.

Computing the evaluations for $i=0,1,2$ allows $u$ to verify the uniqueness of its color among its neighbors.
However, $u$ also has to verify that $P_u$ is consistent with the coloring given by Merlin.
That is, $P_u$ was defined as
$P_u=\sum_{v\in N(u)} C_u C_v$,
but the evaluations are of
$P_u=P^{(0)}_u+\sum_{v\in N(u)}H_v$,
and $u$ has to verify that 
these two polynomials indeed coincide.To this end, $u$ computes $P_u$ in a randomly chosen point $i^*$, with the help of its neighbors, and verify these are consistent.
Concretely, $u$ collects the evaluations of $C_v(i^*)$ from all neighbors $v$ to compute $\sum_{v\in N(u)} C_u(i^*)C_v(i^*)$, 
collects the evaluations of $H_v(i^*)$ to compute
$P^{(0)}_u(i^*)+\sum_{v\in N(u)}H_v(i^*)$, and verify there two are equal.

The technical details, and especially guaranteeing and proving that this is implemented without leaking information, are somewhat more involved, and described in \cref{sec:threecol-polynom}.
Altogether, this results in a more communication-efficient protocol for $3$-colorability with improved soundness.
In fact, the protocol can easily be extended to guarantee soundness $O(1/n^\delta)$ for any positive constant $\delta$ and for $c$-colorability for any $c\in[2,n]$, with no asymptotic cost (see \cref{claim:kcol and s soundnes} for details).
Moreover, the protocol uses public randomness shared among all nodes. This can be replaced by private randomness at the cost of one additional communication round between nodes, but without requiring extra prover communication. 
Finally, we show that this protocol does not require ID assignment for the nodes or initial knowledge of their neighbors. 
Each node only needs to know its number of neighbors but does not need to know its own ID or those of its neighbors.

\subsubsection{Triangle-freeness and Certificate-Communication Trade-off}

In \cref{sec:triangle-freeness}, we turn to the problem of triangle freeness, where the goal is to certify that the communication graph contains no triangles. The core of our approach is a polynomial $P_u$ for each node $u$, which is similar to the polynomial $P_u$ from the previous section. 
The polynomial $P_u$ describes the neighbors of $u$'s neighbors, i.e., its distance-2 neighborhood. The node $u$ must verify both that $P_u$ describes no triangles and that it correctly represents the neighborhood of $u$. This is achieved using polynomial sharing, but in a more sophisticated manner. With this idea, we are able to show the following.

\begin{restatable*}{thm}{thmtrifree}
	\label{thm:trifree}
	For every $3\leq\alpha \leq \sqrt n$ we have $\tfree \in \dzma\left(\frac{n}{\alpha}\log n,\; \alpha\log n\right)$.
\end{restatable*}

A straightforward sharing of $P_u$ as in the colorability protocol would require linear-size messages from the prover.
However, we establish a trade-off inspired by prior work on (non zero-knowledge) triangle-freeness certification~\cite{CrescenziFP19}. 
For a parameter $\alpha$ between $3$ and $\sqrt{n}$, we reduce the communication from Merlin to Arthur from $n \log n$ to $(n/\alpha) \log n$, at the expenses of increasing the number of bits exchanged between neighboring nodes from $\log n$ to $\alpha \log n$.

Once again, this result can be extended to improve soundness and to accommodate private randomness. 
More interestingly, the protocol also works without nodes and their neighbors having assigned IDs. 
If the maximum degree $\Delta$ satisfies $\Delta = o(n^{1/3})$, we further reduce the communication to only $O(\Delta^3/\alpha \log \Delta)$ for node certificates and $O(\alpha \log \Delta)$ for messages exchanged between nodes. This improves upon the previous work~\cite{CrescenziFP19} while incorporating the zero-knowledge property.

\subsubsection{A General Compiler}
\label{sec:intro-compiler}
Finally, in \Cref{sec:compiler}, we show a {\em compiler} that gives a $\dzma$ protocol for every graph property in ${\sf NP}$. Our protocol works in the {\em random oracle model} (ROM), an idealized cryptographic model where all the parties can query an oracle $H$ with an input $x$ and receive as an answer a fixed random value $H(x)$. ROM is a common tool in cryptography, being a convenient tool for proving the security of protocols, and it can usually be replaced by concrete assumptions.~\footnote{See \Cref{sec:ROM} for a more detailed discussion on this model.}

The ROM will be useful for two reasons. First, it is well-known that every language in ${\sf NP}$ admits  (centralized) NIZKs in ROM. Secondly, we can easily implement bit-commitments in the ROM. A bit-commitment is a two-phase protocol that allows the sender to commit to a message $m$ to the receiver without disclosing it. Later, the sender can open the commitment to the receiver, who learns the message $m$. The cryptographic properties that we want from the commitment schemes is that it is hiding, i.e., the receiver does not learn the message before the opening phase, and binding, i.e., the sender cannot open the commitment to a message $m' \ne m$. 

Our universal $\dzma$ protocol uses these two basic building blocks as follows:
\begin{enumerate}
    \item The prover sends the commitments to the adjacency matrix of the graph to all nodes
    \item For each node $v$, the prover opens the commitments for their edges
    \item The prover sends a centralized NIZK that the committed graph has the property of interest
    \item Each node then checks that
    \begin{enumerate}
        \item the prover sent the same commitment to all the nodes via an equality sub-protocol
        \item the opened edges are consistent with the node's neighborhood
        \item the NIZK verification passes.
    \end{enumerate}
\end{enumerate}
We stress that lines (1)--(3) above are sent in a single message from the prover to the nodes, and require no interaction.

Completeness is straightforward, and the soundness property comes from the centralized NIZK and the equality sub-protocol, along with the binding property of the commitment scheme. Using the zero-knowledge of the centralized NIZK and the hiding property of the commitment scheme, we can construct a simulator, and prove the zero-knowledge property of our distributed protocol through a sequence of hybrids. 

While the amount of communication with the neighbors can be easily bounded by $O(\log n)$, measuring the communication with the prover is more subtle. We know that every problem in ${\sf NP}$ admits a NIZK, but we can only bound its size by $\poly(n)$, which can vary from problem to problem. In particular, this value depends on the blow-up of the instance/proof size when we consider the reduction from the graph property of interest to an NP-complete problem.

\begin{restatable*}{thm}{thmuniversal}\label{thm:universal-zk}
	Any graph property in NP is also in $\dzma(\poly n, \log n)$.
\end{restatable*}

An interesting feature of our universal $\dzma$ protocol is that the zero-knowledge property holds against a coalition of an arbitrary set of nodes. Moreover, we can achieve negligible soundness by increasing the communication between the neighbors to polylogarithmic.

We notice that \cite{BickKO22} proposes a general compiler from PLS to distributed zero-knowledge interactive proofs. Their compiler leads to a protocol where the communication with the prover would be much more efficient compared to ours for graph properties that have a short PLS but heavy NIZK. 
On the other hand,  our protocol is more efficient in the number of rounds and bits of communication between neighbors, and achieves stronger zero-knowledge security. We give a more detailed comparison in \Cref{sec:discussion-compiler}.  %

\subsection{Related Work}
Our work continues a long line of research on distributed certification, where a prover assigns certificates for some graph property to the nodes, who then verify the property in a deterministic manner (proof-labeling schemes (PLFs)~\cite{KKP10}, locally checkable proofs~\cite{GS16}, and the class $\mathsf{NLD}$~\cite{FKP13}).
More relevant to our work are RPLS~\cite{FraigniaudPP19}, where the nodes   exchange randomized messages instead of deterministic ones.
Our work is also related to distributed interactive proofs~\cite{KOS18}, where the units can exchange messages with the prover instead of merely receiving certificates. 
Specifically, we elaborate upon a previous dMA protocol for triangle-freeness~\cite{CrescenziFP19}.

As previously discussed, the most relevant related work to our paper is \cite{BickKO22}, who first defined the notion of distributed zero-knowledge. 
Their most general definition is of \emph{distributed knowledge}, denoted dK[$r,\ell,\mathcal{A}_v,\mathcal{A}_s,k$],
where the zero-knowledge property is w.r.t. a distributed protocol from some family.
Here, $r$ is the number of communication rounds between the prover and the nodes, and $\ell$ is the number of bits exchanged in each such  round. 
$\mathcal{A}_v$ is the type of distributed algorithm that the nodes run after communicating with the prover, in order to decide whether to accept or reject, while
$\mathcal{A}_s$ is the type of distributed algorithm used by the simulator. 
Finally, $k$ is the size of an adversarial coalition of nodes for which the protocol is still zero-knowledge.
Let $\cgst(\mu)$ be the the standard $\cgst$ model with $O(\mu)$-bit massages,
and $\bot$ be the class of zero-rounds, no-communication protocols.
Then, $\dzma(\ell,\mu)=$dK[$1,\ell,\cgst(\mu),\bot,1$], which is also denoted dSZK[$1,\ell,\cgst(\mu),1$] there.

Finally, we notice that~\cite{BickKO22} use the classiacl Shamir secret sharing~\cite{Shamir79} in their zero-knowledge protocol for spanning tree verification. 
Our polynomial sharing technique scheme extends Shamir secret sharing, exploiting more structured polynomials to enable efficient verification without revealing further information.

\subsection{Open Questions}

\myparagraph{Polynomial sharing} Our protocols for coloring and triangle-freeness rely on the technique of polynomial sharing. We leave as an open problem if such a technique can be used to devise zero-knowledge protocols for a more general family of problems.

\myparagraph{Zero-knowledge against coalition} We leave as an open problem showing protocols for coloring and triangle-freeness against a coalition of multiple malicious nodes. We notice that in order to achieve such a result for coloring, we need a zero-knowledge protocol where the parties actually do not learn the color assigned by the prover explicitly. It is very intriguing if one can certify proper coloring without given the nodes their colors (or some function that reveals information on the colors). Note that this an independent graph-theoretical question, not necessarily connected to ZK.

\myparagraph{Complete trade-off on communication for triangle-freeness} As described in \Cref{sec:our-results}, we achieve a $\dzma$ protocol with a trade-off $\frac{n}{\alpha} + \alpha$  vs. $\alpha$ between the communication with the prover and with the neighbors. Since in the (non-zero-knowledge) PLS setting, we have the trade-off $\frac{n}{\alpha}$  vs. $\alpha$, we leave as an open question showing a zero-knowledge protocol with better communication for $\alpha > \sqrt{n}$, or proving its impossibility.

\myparagraph{Necessity of computational assumptions for universal protocol in the plain model} In our compiler, we make use of the Random Oracle to achieve our universal zero-knowledge protocol. We leave as an open question if distributed non-interactive zero-knowledge can be implemented in the plain model unconditionally, or if it impossible as in the centralized setting.

\section{Preliminaries}
\label{sec:prelim}

\subsection{Notations}

We consider a network represented by its communication graph $G=(V,E)$, which is simple, connected and undirected.
It consists of $|V|=n$ computational units (nodes) that communicate synchronously using $|E|=m$ communication links (edges).
We use $(u,v)$ to denote an (undirected) edge, 
$N(u)$ to denote the set of neighbors of $u\in V$, and 
$\deg(u)=|N(u)|$ for its degree.
We denote by $\Delta$ the maximum degree of a node in the graph (or more precisely, an upper bound on it).

A configuration $\con$, consists of a graph $G = (V, E)$, a state
space $S$, and a state assignment function $s: V\rightarrow S$. The state of
a node $v$, denoted $s(v)$, includes all local input to v. In particular,
the state always includes the node unique identifier $\ID(v)$, where $\ID : V \rightarrow \{1,\ldots, n^c\}$ is a one-to-one function and $c \ge 1$ is a constant. The state may include weights of incident edges (for edge-weighted networks) and other data like, e.g., the
result of an algorithm.
A languages $\mL$ is a set of configurations.
Each node initially knows its state and the IDs of its neighbors (but not their entire states).
In addition, all nodes have access to a source of shared randomness $r$, which is unknown to the prover.

For a positive integer $k$, we denote $[k]=\{0,\ldots,k-1\}$.
We say that a function $f$ is negligible, i.e. $f(x) = \negl(x)$, if for every fixed constant $c \in \mathbb{N}$, $f = o(1/n^c)$.

\subsection{Distributed Non-Interactive Zero-Knowledge Proofs}
\label{sec:prelim-definition}
We start by defining Distributed non-interactive (non zero-knowledge) proofs.

\begin{definition}[Distributed non-interactive proofs]\label{def:ni-proofs}
    The class of Distributed non-interactive proofs $\dma[\ell, \mu]$ contains the languages $\mL$ for which there exists
	a protocol of the following format:
    \begin{enumerate}
        \item The prover sends a message $\sigma_v$ to every node $v \in V$ of size at most $\ell$
        \item Based on its state and the randomness $r$ shared by the nodes (and unknown to the prover), each node $v$ sends a message  $\gamma_{v,u}$ of size at most $\mu$ to each neighbor $u \in N(v)$. We define $\gamma_u$ to be the sequence of messages received by node $u$ from its neighbors in a canonical order;
        \item Node $v \in V$ decides to accept or reject based on $(s(v),r, \sigma_v, \gamma_v)$.
    \end{enumerate} 

    Moreover, we require the following properties
	\begin{itemize}
		\item \textbf{Completeness:} If $\con \in \mL$,  there are messages $(\sigma_v)_{v \in V}$ that make the nodes accept with probability $1$.
		\item \textbf{Soundness:} If $\con \not\in \mL$,  for every $(\sigma_v)_{v \in V}$, at least one node reject with probability at least $1 - O(1/n)$.
    \end{itemize}
\end{definition}

Our definition of distributed non-interactive proofs can be seen as a \emph{one-sided randomized proof-labeling scheme (RPLS)}~\cite{FraigniaudPP19} with shared randomness between the nodes.

Given a Distributed non-interactive proof, we define the view of a node.

\begin{definition}[View of $v \in V$]\label{def:view}
In the execution of a distributed non-interactive proof, the {\em view} of a node $v \in V$, defined by a random variable $VIEW_v$ that contains the values $(s(v),r,\sigma_v,\gamma_v)$ described in \Cref{def:ni-proofs}. We denote $\mathcal{VIEW}_v$ as the distribution of $VIEW_v$. 
\end{definition}

Now we turn to define the notion of Distributed non-interactive zero-knowledge proofs.

\begin{definition}[Distributed non-interactive zero-knowledge proofs] \label{def:dnizk}
    The class of distributed non-interactive zero-knowledge proofs $\dzma[\ell, \mu]$ contains the languages $\mL$ for which there exists a distributed non-interactive proof for $\mL$ with the extra property:
	\begin{itemize}
        \item \textbf{Zero-knowledge:} There exists a randomized algorithm ${\sf Sim}$ that every node can run locally (without any communication), such that if $\con \in \mL$, the distribution of ${\sf Sim}$'s output is equal to $\mathcal{VIEW}_v$.
	\end{itemize}
\end{definition}

\begin{remark}
In \cite{BickKO22}, their definition of distributed strong ZK has two extra parameters: the number of rounds with the prover and the maximum number of malicious nodes. We notice that while the former is not necessary in our definition (since there is always a single message from the prover), we can also extend the notion of zero-knowledge for a coalition of malicious parties. We discuss this extension in \Cref{sec:compiler}.     
\end{remark}
\begin{remark}
We can also discuss relaxed versions of zero-knowledge, where we require that the distribution of the views of the nodes to be (statistically or computationally) indistinguishable from the output of the simulator. We also discuss more about this extension in \Cref{sec:compiler}.
\end{remark}

\section{A $\dzma$ Protocol for $3$-colorability}
\label{sec:threecol-polynom}

In this section, we present a $\dzma$ proof for $3$-coloring, i.e., the problem of coloring a graph's nodes with $3$ colors, such that no two neighbors have the same color.
The rest of this section is dedicated to proving the following theorem and its extensions.

\thmcoloring

\subsection{The $3$-colorability Protocol}

    Let $q$ be a prime number in $\{n+1,...,2n\}$, and $\fq$ be the finite field with $q$ elements;
    $n$ and $q$ are known to all the parties.
    
    Merlin chooses a $3$-coloring, and applies a random permutation on colors, resulting in the  assignment of a color $\col(u)$ to every node $u$.
    
    For each node $u$, Merlin chooses a uniformly random field element $r_u\in\fq$ and sends it to the $u$.
    Let $C_u:\fq\to\fq$ 
    be the \emph{coloring polynomial} of $u$, defined as the unique polynomial of degree $\deg(C_u)\leq3$ satisfying:
    \begin{align}\label{eq:polynomial}
    C_u(i)=
    \begin{cases}
    	0, & \text{if $0\leq i\leq2$ and $i\neq \col(u)$}\\
    	1, & \text{if $i= \col(u)$}\\
    	r_u, & \text{if $i= 3$ .}
    \end{cases}
    \end{align}
    The values of $C_u$ for $3<i<q$ are uniquely determined due to the degree bound.

    Observe that if two neighboring nodes $u,v$ have different colors, then for each $0\leq i\leq 2$
    either $C_u(i)=0$ or $C_v(i)=0$, so $C_u(i)C_v(i)=0$ for all $0\leq i\leq 2$.
    For each node $u$ we define its polynomial $P_u$ by 
    \[P_u=\sum_{v\in N(u)} C_u C_v,
    \]
    and note that $\deg(P_u)\leq 6$.
    Moreover, for $0\leq i\leq 2$, the value $P_u(i)$ is a sum of at most $\Delta<q$ values which are all in~$\{0,1\}$, and hence for $0\leq i\leq 2$ we have $P_u(i)=0$ if and only if $C_u(i) C_v(i)=0$ for all $v\in N(u)$.
    Hence, to verify the coloring around $u$, it is enough to verify that $P_u(i)=0$ for $i\in\{0,1,2\}$.
    
    The polynomial $P_u$ depends on the colors of the neighbors of $u$ and their random values $r_v$, and hence it is not known to node $u$.
    If we were not in the zero-knowledge setting, the prover could have given $P_u$ to $u$ and let it verify it nullifies on $0,1,2$.
    However, $u$ should not learn $P_u$, as it might disclose information on the neighbors' colors; one might suspect that the random values $r_v$ are enough to conceal any useful information encoded in $P_u$, so let us falsify this intuition.
    Define 
    $P_{N(u)}=\sum_{v\in N(u)} C_v$
    and note that 
    $P_u=C_u P_{N(u)}$.
    If $u$ would have known both $C_u$ and $P_u$, it could have easily computed $P_{N(u)}$ by polynomial division.
    However, for $i\in\{0,1,2\}$, $P_{N(u)}(i)$ indicates how many of $u$'s neighbors are colored $i$, which leaks information on the graph.
    For a concrete example, assume $\deg(u)=2$ and moreover $P_{N(u)}(i)=2$ for some $i\in\{0,1,2\}$. Then, $u$ learns its neighbors are both colored $i$, so they are not neighbors of one another.
    
    Hence, we need to verify the values of $P_u$ on $i\in\{0,1,2\}$ without having any node learning $P_u$. 
    To this end, we ``split''  $P_u$ among $u$ and all its neighbors, à la secret-sharing.
    
    Each node $u$ receive one \emph{helper polynomial} $H_u$ over $\fq$, which has degree at most $6$ and whose coefficients are chosen u.a.r.
    For a node $u$, the helper polynomials $H_v$ of all its neighbors will serve as shares of $P_u$, 
    and $u$ itself will receive one extra share $P^{(0)}_u$ which is also a polynomial over $\fq$ of degree at most $6$, 
    such that
    \[
    P_u=P^{(0)}_u+\sum_{v\in N(u)}H_v.
    \]
    We then use the linearity of polynomials (in their coefficients) to allow $u$ evaluating $P_u$ at a point $i$ without learning $P_u$:
    to compute $P_u(i)$, 
    node $u$ computes $P^{(0)}_u(i)$,
    each neighbor $v$ of $u$ computes 
    $H_v(i)$ and sends it to $u$,
    and then $u$ can retrieve $P_i(i)$ by
    $P_u(i)=P^{(0)}_u(i)+\sum_{v\in N(u)}H_v(i).$

    \cref{prot:three-col-poly-merlin} below shows the certificate assignment by Merlin.
    It first permutes the colors to get a proper $3$-coloring~$\col$, then chooses a random value $r_u$ for each node $u$,
    from which it computes $C_u$
    and $P_u$ for every $u$,
    as described above.
    Then, it creates the random helper polynomials $H_u$ for each node, and uses them to compute the local share $P^{(0)}_u$ of each $u$.
    Finally, it sends $\col(u), r_u$ and the polynomials $P^{(0)}_u$ and $H_u$ to each node $u$.
    
    To verify the given $3$-coloring (and thus decide $3$-colorability), each node $u$ follows \cref{prot:three-col-poly}.
    Its first task is to check that the polynomial defined as the sum of polynomials $P^{(0)}_u$ and $\{H_v\mid v\in N(u)\}$ describes a proper $3$-coloring, in the sense that it nullifies on $i\in\{0,1,2\}$.
    This is done using the evaluations of $P^{(0)}_u(i)$ on $i\in\{0,1,2\}$ made by $u$ itself, and by the evaluations of $H_v(i)$ on the same values of $i$ sent to it by each neighbor $v$.
    
    The second task of $u$ is to make sure that the above polynomial is identical to the coloring polynomial, i.e., the polynomial induced by the color $\col(u)$, the random value $r_u$, and the colors and random values of the neighbors.
    To this end, all the nodes pick a random point $i^\ast$ together, using the shared randomness, 
    and make sure both polynomials (the sum of shares and the coloring polynomial) agree on $i^\ast$.
    This requires them to exchange their evaluations of $H_u$ on this point (the shares), as well as their evaluations of the coloring polynomials $C_u$ on it.

We start with Merlin's protocol.

    \begin{protocol}[$3$-colorability low-communication protocol]\label{prot:three-col-poly-merlin}
    Protocol for \textbf{Merlin}
    \begin{enumerate}[label=(M\arabic*)]
    \item 
    Apply a random permutation on the coloring to get a proper $3$-coloring $\col:V\to\{0,1,2\}$
    \item
    For each $u\in V$,
    pick $r_u\in\fq$ u.a.r. 
    
    \item 
    For each $u\in V$, use $\col(u)$ and $r_u$ to compute $C_u$ by \cref{eq:polynomial}
    \item 
    For each $u\in V$,
    let $P_u\gets \sum_{v\in N(u)}C_u C_v$
    \item 
    For each $u\in V$,
    pick~$7$ coefficients in $\fq$ independently and u.a.r. to create a polynomial $H_u$ with $\deg(H_u)\leq 6$
    \item 
    For each $u\in V$,
    let $P^{(0)}_u\gets P_u-\sum_{v\in N(u)}H_v$

    \item\label{P3colmLsend}
    For each $u\in V$,
    send $\col(u), r_u$ 
    the coefficients of $P^{(0)}_u$ and
    the coefficients of $H_u$ 
    to $u$
    \end{enumerate}
    \end{protocol}

We move on to Arthur, i.e., node~$u$'s protocol.
\begin{protocol}
    [$3$-colorability low-communication protocol]\label{prot:three-col-poly}
    Protocol for \textbf{node $u$}
    \begin{enumerate}
    \item\label{P3colaLreceive}
    Receive $\col(u)$, $r_u$,    the coefficients of $P^{(0)}_u$ and
    the coefficients of $H_u$ 
    from Merlin

    \item 
    Pick $i^\ast\in [q]\setminus[3]$ u.a.r. using the shared randomness
    
    \item\label{P3colaLsendHis} Send the evaluations 
    $H_u(i)$, for all $i\in\{0,1,2,i^\ast\}$,
    to every neighbor $v\in N(u)$
    
    \item
    From every $v\in N(u)$ 
    receive the evaluations 
    $H_v(i)$ for all $i\in\{0,1,2,i^\ast\}$
    
    \item 
    Verify that 
    $P^{(0)}_u(i)+
    \sum_{v\in N(u)}H_v(i)=0$ 
    for all $i\in\{0,1,2\}$, otherwise \textbf{reject}
    
    \item 
    Using $\col(u)$ and $r_u$,
    construct the polynomial
    $C_u:\fq\to\fq$ by  \cref{eq:polynomial}
    
    \item\label{P3colaLsendCi}
    Send $C_u(i^\ast)$
    to all your neighbors
    \item
    Receive $C_v(i^\ast)$
    from each neighbor~$v$
    
    \item 
    Verify that  
    $\sum_{v\in N(u)}C_u(i^\ast) C_v(i^\ast)
    =
    P^{(0)}_u(i^\ast)+
    \sum_{v\in N(u)}H_v(i^\ast)$, 
    otherwise \textbf{reject}
    \item \textbf{Accept}
    \end{enumerate}
    \end{protocol}

Completeness and soundness proofs are inspired by~\cite{CrescenziFP19}.
Intuitively, the protocol is zero-knowledge since all the marginal distribution of the information that each node receives is uniformly random. We can formalize the proof of this property by showing a simulator that indeed only feeds random values to a malicious node, conditioned on the fact that these random values pass the consistency checks. Using a {\em hybrid argument}, a standard proof technique in cryptography, we can show that the distribution of the real protocol and the simulated one are exactly the same.

\subsection{Analysis of the $3$-colorability Protocol}
We now prove these properties in more detail.
    
    \subsubsection*{Completeness}
    The completeness of the scheme is immediate --- if the graph has a $3$-coloring and everyone follow the protocol, then all nodes accept.

    \subsubsection*{Soundness}
    For soundness, note that if the graph has no $3$-coloring, than for every color assignment $\col$ there are two neighbors $u$ and $v$ such that $\col(u) = \col(v)$.

    For each node $u\in V$, let $\tilde{P}^{(0)}_{u}$ and $\tilde{H}_{u}$ be the polynomials provided by the prover to $u$. 
    Let
    $\tilde{P}_{u} = \tilde{P}^{(0)}_{u} +
    \sum_{v\in N(u)} \tilde{H}_{v}$.
    If $\tilde{P}_u(i) \ne 0$ for some $i \in \{0,1,2\}$, then node $u$ rejects with probability $1$.
    Assume this is not the case, i.e., $\tilde{P}_u(i) = 0$ for all $i \in \{0,1,2\}$ and all $u\in V$. 
    Then, node $u$ accepts iff
    \begin{align}\label{eq:equality-polynomials}
    \tilde{P}_{u}(i^\ast) = \sum_{v\in N(u)} C_u(i^\ast)C_v(i^\ast).
    \end{align}
    Here, $C_u$ is the polynomial induced by the values $\col(u)$ and $r_u$ sent to $u$ by Merlin, who might have chosen them not by to the protocol.

    Let $u$ be a node colored  $\col(u)=i'\in\{0,1,2\}$ that is not properly colored, i.e., with at least one neighbor $v\in N(u)$ satisfying  $\col(v)=i'$.
    The two polynomials $\tilde{P}_{u}$ and $\sum_{v \in N(u)} C_uC_v$ are different since $\tilde{P}(i) = 0$ for all $i \in \{0,1,2\}$,
    while $0 < \sum_{v \in N(u)} C_u(i')C_v(i') < n$ (this is true as a sum over $\mathbb{Z}$ and over $\fq$). 
    Since both polynomials have degree at most $6$ the polynomial 
    $\tilde{P}_{u} - \sum_{v\in N(u)} C_uC_v$
    has at most $6$ distinct roots. Therefore, there exists at most $6$ values of $i^\ast$ such that \Cref{eq:equality-polynomials} holds. Since $i^\ast$ is picked uniformly at random in $[q]\setminus[3]$, the probability that \Cref{eq:equality-polynomials} holds at node $u$ and accepts is at most $\frac{6}{q-3} < \frac{6}{n-3}$.

    \subsubsection*{Communication}
    The communication between Merlin the Arthur (the node) only happens in Line~\ref{P3colmLsend} of Protocol~\ref{prot:three-col-poly-merlin}, which corresponds to Line~(\ref{P3colaLreceive}) of Protocol~\ref{prot:three-col-poly}.
    For each node, Merlin sends two field elements and two sets of six coefficients in the field, for a total of $O(1)$ field elements. This requires $O(\log q)=O(\log n)$ bits of communication.
    
    Messages are sent between graph nodes in Line~(\ref{P3colaLsendHis}) and Line~(\ref{P3colaLsendCi})  of Protocol~\ref{prot:three-col-poly}, and are received in the following lines.
    Here, neighbors exchange four field elements,
    for a total of $O(\log q)=O(\log n)$ bits per edge.
    
\subsection{The Zero-Knowledge Property of the $3$-colorability Protocol}

The general idea of the simulator is to replace all the communication received by a malicious node by random polynomials/values conditioned on the fact that such values pass the test performed by the node.
All random choices below are uniform from $\fq$ and independent, except for Line~\ref{P3colsLpicki123} and Line~\ref{P3colsLpickiast}, which are explained below.

    \begin{protocol}
    \label{prot:sim-three-col-poly}
    Simulator for malicious node $u$
    
    \begin{enumerate}
    \item Pick ``shared'' randomness $i^\ast$
    
    \item Pick random $c_u$,  $r_u$, $7$ coefficients that define a degree-$6$ polynomial $H_u$, and
    $7$ coefficients that define a degree-$6$ polynomial $\tilde P^{(0)}_u$

    \item \label{P3colsLpicki123}
    For each $i \in {0,1,2}$, pick random values $\tilde{h}_v^i$ for all $v\in N(u)$ such that $\tilde P_u^{(0)}(i) + \sum_{v\in N(u)} \tilde{h}_v^i = 0$

    \item 
    \label{P3colsLpickiast}
    Pick random values $\tilde{c}_v, \tilde{h}_v^{i^\ast}$ for all $v\in N(u)$ such that $\sum_{v \in N(u)} C_u(i^\ast)\tilde{c}_v= \tilde P^{(0)}_u(i^\ast) + \sum_{v\in N(u)} \tilde{h}_v^{i^\ast}$

    \item Simulate the message with $c_u, r_u$ and the coefficients of $H_u$ and $P^{(0)}_u$ from the prover to $u$
    
    \item Simulate the message $(\tilde{h}_v^i)_{i \in \{0,1,2,i^*\}}, \tilde{c}_v$ from each node $v \in N(u)$ to $u$

    \end{enumerate}
    \end{protocol}
In Line~\ref{P3colsLpicki123}, for each $i\in\{0,1,2\}$ the values $\tilde{h}_v^i$ are not independent, as they must satisfy the specified equation.
To choose them, we arbitrarily pick one neighbor $v'\in N(u)$,
choose $\tilde{h}_v^i\in\fq$ independently and u.a.r. for all $v\in N(u)\setminus \{v'\}$, and then set $\tilde{h}_{v'}^i=-( \tilde P_u^{(0)}(i) + \sum_{v\in N(u)\setminus\{v'\}} \tilde{h}_v^i )$.
Note that $\tilde{h}_{v'}^i$ is also distributed uniformly in $\fq$: it is enough that one of the above summands is chosen u.a.r., and this is always the case since $\tilde P_u^{(0)}(i)$ always exists.
Line~\ref{P3colsLpickiast} is similar: all the values except for $\tilde{h}_{v'}^{i^\ast}$ for some $v'\in N(u)$ are chosen independently (including $\tilde{c}_{v'}$) and 
$\tilde{h}_{v'}^{i^\ast}$ is then chosen to complement them.

Our goal now is to prove that from the perspective of the malicious node, the distribution of transcripts that it receives is exactly the same in the real run of the protocol and in the simulated version.
To this end, we use the standard technique in cryptography of providing {\em hybrids}, intermediate protocols between the real run and the simulated one. 
We consider a series of steps, where the first step is the real run of the protocol, the last step is the simulated version, and we can show that each consecutive pair of steps is indistinguishable. 
By transitivity, we have that the real run and the simulated one are indistinguishable.

  \begin{description}
    \item[Hybrid 0:] The transcript of the real run of \Cref{prot:three-col-poly}
    \item[Hybrid 1:] The same as hybrid $0$, but instead of the neighbors sending $C_v(i^\ast)$, node $u$ receives random values $\tilde c_v$ such that $P^{(0)}_u(i^\ast) + \sum_{v\in N(u)} H_v(i^\ast) = \sum C_u(i^\ast) \tilde c_v$
    \item[Hybrid 2:] The same as hybrid $1$, but instead of the neighbors sending $H_v(i)$, node $u$ receives random values $\tilde h^{i^\ast}_v$ such that $P^{(0)}_u(i^\ast) + \sum_v \tilde  h^{i^\ast}_v= \sum C_u(i^\ast) \tilde c_v$ and for $i \in \{0,1,2\}$, $P^0_u(i) + \sum_{v\in N(u)} \tilde{h}^i_v = 0$
    \item[Hybrid 3:] Same as hybrid $2$, but instead of receiving the honest $P^{(0)}_u$, node $u$ receives a random polynomial $\tilde{P}^{(0)}_u$
  \end{description}

\begin{lemma}
  The distribution of transcripts between Hybrid $0$ and Hybrid $1$ are identical.
\end{lemma}
\begin{proof}
  In order to prove this lemma, we use the fact that for for every $i > 2$, $C_v(i)$ is a uniformly random value. More precisely, we show that for every $y \in \mathbb{F}_q$, 
      \begin{align}\label{eq:uniform-evalution}
          \Pr_{r_u}[C_v(i) = y] = \frac{1}{q}.
      \end{align} 
    Given \Cref{eq:uniform-evalution}, the distribution of the values $C_v(i^\ast)$ in the real protocol consists of an uniformly random value such that $P^{(0)}_u(i^\ast) + \sum_{v\in N(u)} H_v(i^\ast) = \sum C_u(i^\ast) C_v(i^\ast)$, which is the exactly same distribution of Hybrid~$1$.
      
      We now show now that \Cref{eq:uniform-evalution} holds.
      Since the values of $C_v(0)$, $C_v(1)$ and $C_v(2)$ are fixed and $C_v(\cdot)$ has degree at most $3$, each choice of $C_v(3)=r_v \in \mathbb{F}_q$ 
      induces a unique polynomial, and hence a unique value of $C_v(i)$ for every $i \geq 2$.
      On the other hand, for each $i \geq 2$, each choice of $C_v(i)$ induces a unique polynomial and a unique value of $r_v$.
      Therefore, for every fixed $i\geq 2$ there is a bijection between $r_v$ and $C_v(i)$ and \Cref{eq:uniform-evalution} holds.
  \end{proof}

  \begin{lemma}
      The distribution of transcripts between Hybrid $1$ and Hybrid $2$ are identical.
  \end{lemma}
  \begin{proof}
      Given that each degree-$6$ polynomials $H_v$ are picked uniformly at random, the distribution of $$\left((H_v(0),H_v(1),H_v(2),H_v(i^\ast))\right)_{v \in N(u)}$$ is uniformly random over $\mathbb{F}_q^{4|N(u)|}$ with the constraint that $P^{(0)}_u(i^\ast) + \sum_{v\in N(u)} H_v(i^\ast) = \sum_{v\in N(u)} C_u(i^\ast) \tilde c_v$. This is exactly the same distribution as Hybrid $2$.
  \end{proof}

  \begin{lemma}
      The distribution of transcripts between Hybrid $2$ and Hybrid $3$ are identical.
  \end{lemma}
  \begin{proof}
      We have that $\sum_{v \in N(u)} H_v$ is also a random polynomial of degree at most $6$. So it follows that $P^{(0)}_u = P_u - \sum_{v \in N(u)} H_v$ is also a random polynomial with degree at most $6$ under the conditioned on that the two checks pass.  In Hybrid $3$, we pick $P^{(0)}_u$ from the exact same distribution.
  \end{proof}
This concludes the proof of Theorem~\ref{thm:3col}.

\subsection{Discussion and Extensions of  the $3$-colorability Protocol}
\paragraph{More colors, less errors}
The above protocol proves $3$-colorability with soundness error of roughly $1/n$.
However, it is rather simple to extend it to accommodate $c$-coloring for every $c\in [2,n]$ and achieve soundness $s$ for every $s\in(0,6/(n-3)]$.
To achieve these goals, 
we only need to adjust the degrees of the polynomials and the size $q$ of the field $\fq$.

For $c$-coloring, we change \cref{eq:polynomial} to have $C_u(0),...,C_u(c-1)$ indicating $\col(u)$, and $C_u(c)$ random.
The polynomial $C_u$ has degree $c$, and $P_u^{(0)}$ and $H_u$ have degrees $2c$.
The soundness is then $\frac{2c}{q-c}$ (the proof is identical), and getting an arbitrary soundness error $s$ requires setting
$3c/s<q\leq 6c/ s$.
This choice of parameters guarantees $c\leq n=O(1/s)$, and thus $\log q=O(\log(1/s))$ communication.

\begin{claim}
	\label{claim:kcol and s soundnes}
For $2\leq c\leq n$ and a parameter $0<s\in O(1/n)$, 
 $c$-col $\in \dzma\left(\log (1/s),\; \log (1/s)\right)$ with soundness error $s$.
\end{claim}

\paragraph{Private randomness}
We presented a public randomness protocol, where the randomness was used to pick $i^\ast$.
Our protocol also extends to the private randomness case, at the expense of two communication rounds between neighbors instead of one, and without any additional communication with the prover.
To this end, each node $u$ chooses a uniformly random point $i^\ast_u$ instead of the common $i^\ast$, and verifies \cref{eq:equality-polynomials} at this point.
To this end, $u$ sends $i^\ast_u$ to each neighbor $v\in N(u)$, which answers with $H_v(i^\ast_u)$.
The protocols, proofs, and equations are almost unchanged, except for replacing~$i^\ast$ by~$i^\ast_u$.

\paragraph{IDs and knowledge of $n$}
As a final remark, we note that our protocol also works in the $KT_0$ model, where processes do not initially know the IDs of their neighbors, but only use port numbers; in this case, they shell also not learn the neighbors' IDs. 
In addition, knowledge of $n$ is not necessary, but only a common upper bound on it.
In fact, the protocol can also go through without an upper bound on $n$ if only a common soundness threshold is given ---  the only place $n$ is needed is for setting the field size~$q$, and this can be done as a function of $s$ instead.

\section{A $\dzma$ Protocol for Triangle Freeness}\label{sec:triangle-freeness}

In this section, we use the same principles of our colorability proof to achieve a triangle-freeness protocol. In a sense, it is based on a triangle-freeness distributed proof of Crescenzi, Fraigniaud and Paz~\cite{CrescenziFP19}, but with extensions that make it $\dzma$.

\subsection{The Triangle-Freeness Protocol}

Recall that a graph $G=(V,E)$ is \emph{triangle-free} if for every three nodes $v_1,v_2,v_3\in V$, at least one of the edges $\{v_1,v_2\},\{v_1,v_3\},\{v_2,v_3\}$ is not in $E$.
We denote by $\tfree{}$ the class of all triangle-free graphs (on a given number~$n$ of nodes).

We present a distributed ZK Merlin-Arthur protocol for triangle-freeness, with a trade-off between the certificate sizes and the message sizes. 
A similar result with similar parameters was proven  in~\cite{CrescenziFP19}, but without the ZK property.
Put differently, we show that the prior protocol can be slightly modified to a ZK one without incurring any asymptotic overhead.
We next prove the following.

\thmtrifree

    Let $\alpha$ be the chosen parameter (if $\alpha<3$, replace it by $3$; asymptotically, this will cause overheads).
    Identify the space of node IDs with $[n/\alpha]\times[\alpha]$,
    i.e., identify each node ID $u\in[n]$ with a pair $(i_u,t_u)\in [n/\alpha]\times[\alpha]$.
    If $n/\alpha$ is not integral, use $\lceil n/\alpha \rceil$ instead, and $n/\lceil n/\alpha \rceil$ instead of $\alpha$ (here and henceforth).
    Choose a prime number $q$ such that $n\alpha <q\leq 2n\alpha$, and denote by $\fq$ the finite field with $q$ elements.
    All these parameters are known to all the parties.

    For each node $u$ and each $t\in [\alpha]$, Merlin chooses a uniformly random element of $\fq$ and sends it to the $u$;
    we denote by~$r_{u,t}$ the element received by $u$.
    For every such $u$ and $t$, 
    let $P_{u,t}:\fq\to\fq$ 
    be the unique polynomial with $\deg(P_{u,t})\leq n/\alpha$ satisfying:
    \begin{align}\label{eq:trifree-peralpha}
   	P_{u,t}(i)=
   	\begin{cases}
   		0, & \text{if $0\leq i<n/\alpha$ and $(i,t)\notin N(u)$}\\
   		1, & \text{if $0\leq i<n/\alpha$ and $(i,t)\in N(u)$}\\
   		r_{u,t}, & \text{if $i=n/\alpha$.}
   	\end{cases}
    \end{align}
    For $n/\alpha<i<q$, the values $P_{u,t}(i)$ are uniquely determined due to the bounded degree.

    Observe that if two neighboring nodes $u,v$ are not part of any triangle, then each node $(i,t)\in [n/a]\times[\alpha]$ is not a neighbor of $u$ or not a neighbor of $v$,
    so either $P_{u,t}(i)=0$ or $P_{v,t}(i)=0$.
    Hence, 
    $P_{u,t}(i)P_{v,t}(i)=0$ for all $t$ and $i$.
    
    For a node $u$, let 
    \begin{align}\label{eq:trifree-main}
    P_u=\sum_{v\in N(u)}\sum_{t\in[\alpha]} P_{u,t}P_{v,t}
	\end{align}
    and note that $\deg(P_u)\leq 2n/\alpha$.
    Moreover, for $0\leq i<n/\alpha$ the value $P_u(i)$ is a sum of at most $n\alpha<q$ values which are all in $\{0,1\}$, and hence $P_u(i)=0$ if and only if $P_{u,t}(i)P_{v,t}(i)=0$ for all $v\in N(u)$ and $t\in[\alpha]$.

    The polynomial $P_u$ can be seen as corresponding to $C_u$ in the colorability protocol, and here as well we refrain from giving it to $u$, and split it using helper polynomials given to its neighbors instead.
    Each node $u$ receive a random \emph{helper polynomial} $H_u$ over $\fq$ with degree at most~$2n/\alpha$. 
    Each node $u$ also receives one polynomial $P^{(0)}_u$ with the same degree bound,
    such that
    \[
    P_u=P^{(0)}_u+\sum_{v\in N(u)}H_v.
    \]
    To compute $P_u(i)$, 
    node $u$ computes $P^{(0)}_u(i)$,
    each neighbor $v$ of $u$ computes 
    $H_v(i)$ and sends it to $u$,
    and then $u$ can retrieve $P_i(i)$ by
    $P_u(i)=P^{(0)}_u(i)+\sum_{v\in N(u)}H_v(i).$
    
    The certificate assignment
    procedure goes as follows.
    For each node $u$ and parameter $t\in [\alpha]$, 
    Marlin chooses $r_{u,t}\in [q]$ u.a.r.
    and computes the polynomial $P_u$ by \cref{eq:trifree-peralpha} and \cref{eq:trifree-main} above.
    It also creates the random helper polynomials $H_u$ for each node,
    and uses these polynomials to compute the local share $P^{(0)}_u$ of each $u$.
    Finally, Marlin sends each node $u$ the vector $(r_{u,t})_{t\in[\alpha]}$ of random values, and the coefficients of the polynomials $P^{(0)}_u$ and $H_u$.

    In the verification phase, each node $u$ first verifies that the polynomial induced by  
    $P^{(0)}_u$ received from Merlin and the helper polynomials $H_v$ of the neighbors sum up together to a polynomial that nullifies  for each $i\in[n/\alpha]$.
    This is simply done by evaluating 
    $P^{(0)}_u(i)$ on all $i\in[n/\alpha]$, receiving $H_v(i)$ from each neighbor $v$ for the same values of $i$, and summing these up for each $i$.
    Then, $u$ verifies that the polynomial described by this sum correctly indicates its neighbors. 
    To this end, it uses 
    the values $(r_{u,t})_{t\in[\alpha]}$ received from Merlin to reconstruct the polynomials $P_{u,t}$,
    receives $H_v(i^\ast)$ and $P_{v,t}(i^\ast)$ from each neighbor $v$ and for all $t\in[\alpha]$, and checks that $\sum_{v\in N(u)}\sum_{t\in[\alpha]}
    P_{u,t}(i^\ast) P_{v,t}(i^\ast)
    =
    P^{(0)}_u(i^\ast)+
    \sum_{v\in N(u)}H_v(i^\ast)$.
    
    We start with Merlin's protocol.

    \begin{protocol}[Triangle-freeness protocol]\label{prot:tri-free-merlin}
    Protocol for \textbf{Merlin} with parameter $\alpha$
    \begin{enumerate}[label=(M\arabic*)]
    \item
    For each $u\in V$ and $t\in[\alpha]$
    pick $r_{u,t}\in\fq$ u.a.r. 
    
    \item 
    For each $u\in V$ and $t\in[\alpha]$
    use $r_{u,t}$ to compute $P_{u,t}$
    
    \item 
    For each $u\in V$,
    let $P_u=\sum_{v\in N(u)}\sum_{t\in[\alpha]} P_{u,t}P_{v,t}$
    
    \item 
    For each $u\in V$,
    pick~$2n/\alpha+1$ coefficients in $\fq$ independently and u.a.r. to create a polynomial $H_u$ with $\deg(H_u)\leq 2n/\alpha$
    \item 
    For each $u\in V$,
    let $P^{(0)}_u\gets P_u-\sum_{v\in N(u)}H_v$

    \item\label{Ptri-free-Lsend}
    For each $u\in V$,
    send $(r_{u,t})_t\in[\alpha]$,
    the coefficients of $P^{(0)}_u$ and those of $H_u$ 
    to $u$
    \end{enumerate}
    \end{protocol}

    And move to the nodes' protocol.
    
    \begin{protocol}[Triangle-freeness protocol]\label{prot:tri-free-nodes}
    Protocol for \textbf{node $u$} with parameter $\alpha$
    \begin{enumerate}
        \item
        Receive $(r_{u,t})_{t\in[\alpha]}$, 
        the coefficients of $P^{(0)}_u$ and
        the coefficients of $H_u$ 
        from Merlin
         
        \item 
        Pick $i^\ast\in [q]\setminus[n/\alpha]$ u.a.r. using the shared randomness
        
        \item\label{Ptri-freeLsendHis} 
        Send the evaluations 
        $H_u(i)$, for all $i\in[n/\alpha]\cup\{i^\ast\}$,
        to every neighbor $v\in N(u)$
        
        \item
        From every $v\in N(u)$ 
        receive the evaluations 
        $H_v(i)$ for all $i\in[n/\alpha]\cup\{i^\ast\}$
        
        \item 
        Verify that 
        $P^{(0)}_u(i)+
        \sum_{v\in N(u)}H_v(i)=0$ 
        for all $i\in[n/\alpha]$, otherwise \textbf{reject}
        
        \item 
        Using $(r_{u,t})_{t\in[\alpha]}$,
        construct the polynomials
        $P_{u,t}$
        
        \item\label{Ptri-freeLsendPut}
        For all $t\in[\alpha]$,
        send  $P_{u,t}(i^\ast)$ to all your neighbors
        
        \item
        Receive $P_{v,t}(i^\ast)$ for all $t\in[\alpha]$
        from each neighbor $v$
        
        \item 
        Verify that  
        $\sum_{v\in N(u)}\sum_{t\in[\alpha]}
        P_{u,t}(i^\ast) P_{v,t}(i^\ast)
        =
        P^{(0)}_u(i^\ast)+
        \sum_{v\in N(u)}H_v(i^\ast)$, 
        otherwise \textbf{reject}
        
        \item \textbf{Accept}
    \end{enumerate}
    \end{protocol}
    
    \subsection{Analysis and Zero-Knowledge of the Triangle-Freeness Protocol}
    
    The completeness of the scheme is immediate --- if the graph is triangle free and everyone follow the protocol, then all nodes accept.
    For soundness, 
    let 
    $Q_u =
    P^{(0)}_u+
    \sum_{v\in N(u)}H_v$, where 
    $P^{(0)}_u$ and $(H_v)_{v\in N(u)}$ are the polynomials given to the corresponding nodes by Merlin.
    
    If the graph has a triangle, than there is a node $u$
    and a values $i\in [n/a]$ such that $P_u(i)\neq 0$.
    If the prover supplies the nodes with $Q_u=P_u$, then $u$ will notice that $Q_u(i)\neq 0$ and reject (deterministically).
    
    Finally, we examine the case where $Q_u\neq P_u$ for some node $u$.
    The polynomial $P_u$ is constructed implicitly, as the sum of products of polynomials that depend on the graph structure and the values $r_{u,t}$ received by $u$, in a way guaranteeing it has degree at most $2n/\alpha$.
    The polynomial $Q_u$ is defined explicitly, by the coefficients received by $u$ and its neighbors, and thus also has degree at most $2n/\alpha$.
    Two non-identical polynomials of such degree can be identical on at most $2n/\alpha$ values, thus the probability
    over the choice of  $i^\ast\in [q]\setminus[n/\alpha]$,
    that $P_u(i^\ast)=Q_u(i^\ast)$ is at most $\frac{2n}{\alpha (q-n/\alpha)}<\frac{2}{\alpha^2-1}\leq\frac{1}{4}$.

    \begin{protocol}
    \label{prot:sim-triangle-freeness}
    Simulator for malicious nodes $u$
    
    \begin{enumerate}
    \item Pick ``shared'' randomness $i^\ast$
        \item For each $t\in[\alpha]$ pick $\tilde r_{u,t}\in\fq$ u.a.r. 
        \item Pick $(2n/\alpha + 1)$ coefficients that define two random degree-$(2n/\alpha)$ polynomials $H_u$ and $\tilde P^{(0)}_u$.

    \item \label{PtrianglefLpicki123}
    For each $i \in [n/\alpha]$, pick random values $\tilde{h}_v^i$ for all $v\in N(u)$, such that $\tilde P_u^{(0)}(i) + \sum_{v\in N(u)} \tilde{h}_v^i = 0$  

    \item Pick random values $\tilde{p}_{v,t}, \tilde{h}_v^{i^\ast}$ for all $v\in N(u)$ and $t \in [n/\alpha]$ such that
    \[ 
   \sum_{v\in N(u)}
        \sum_{t\in[\alpha]}
        P_{u,t}(i^\ast) \tilde p_{v,t} =   
       \tilde P^{(0)}_u(i^\ast)+
        \sum_{v\in N(u)}\tilde{h}_v^{i^\ast}.\]
    \item Simulate the message with $\tilde r_u$ and the coefficients of $H_u$ and $\tilde P^{(0)}_u$ from the prover to $u$
    
    \item Simulate the message $(\tilde{h}_v^i)_{i \in [n/\alpha] \cup \{i^*\}}, \tilde{p}_{v,t}$ from each node $v \in N(u)$ to $u$

    \end{enumerate}
    \end{protocol}

    The proof of zero-knowledge closely follows the one presented in \Cref{sec:threecol-polynom}, and it follows by replacing each message from the prover/neighbors by the simulated ones, one by one, and showing that each modification keeps the distribution of the received messages.

\subsection{Discussion and Extensions of the Triangle-Freeness Protocol}

\paragraph{A protocol without IDs}
Unlike the colorability protocol, in the triangle-freeness protocol IDs are used extensively.
It is hence surprising that one can extend the latter protocol to work without IDs, and get the following.

\begin{claim}
	Assume the degrees in the graph are bounded by $\Delta$, and let $n'=\min\{n,\Delta^3\}$.
	Then, for every $3\leq\alpha \leq \sqrt{n'}$ we have $\tfree \in \dzma\left(\frac{n'}{\alpha}\log n',\; \alpha\log n'\right)$
    against a single party.
	This holds even if the nodes do not have IDs and each node knows only the number of neighbors.
\end{claim}

To achieve this claim, Merlin chooses a distance-$3$ coloring of the graph, i.e., a coloring where no two nodes of distance at most $3$ have the same color. 
It applies a random permutation on the colors, and assigns each node its color by the permuted coloring.
All the parties then execute the above protocols with the colors instead of IDs. 
Note that any graph with degree at most $\Delta$ has a distance-$3$ coloring using $\Delta^3$ colors;
it also has such a coloring using $n$ colors by assigning each node a unique color.

Merlin informs each node of its own color and the colors of its neighbors. 
This can be sent along with the other information it sends in \Cref{prot:tri-free-merlin}, and no extra communication rounds are needed.
These values can be encoded using  $O(\Delta\log n')$ bits, so asymptotically, no overhead is incurred (as $\Delta\leq\sqrt{n'}\leq n'/\alpha$).

For correctness, note that the algorithm verifies that no node $u$ has a neighbor $v$ such that both $u$ and $v$ have a neighbor $x$. 
When the colors of a distance-$3$ coloring instead of IDs, the same condition is checked.
The colors sent above are enough for node $u$ to follow \Cref{prot:tri-free-nodes},
and the completeness follows.
For soundness, note that a non-proper distance-$3$ coloring can only make nodes believe a triangle exists (if two neighbors $u$ and $v$ have each a different neighbor colored $x$) and thus reject even when it does not exist.
However, it cannot cause the acceptance of a graph where a triangle exists.
For this reason, the nodes do not have to check that the given coloring is a proper distance-$3$ coloring.

\paragraph{Improved soundness}
We conclude by observing that the soundness error can easily be decreased, as in the triangle-freeness case. 
For example, by choosing $n^3/\alpha<q\leq 2n^3/\alpha$, we get soundness smaller than $1/n$ with no asymptotic overhead.

\section{Universal $\dzma$ Protocol}
\label{sec:compiler}

In this section, we show that every graph property that can be proven in 
NP can be also proven in
zero-knowledge in the distributed setting against computationally-bounded
malicious adversaries. 

We start by reviewing some basic concepts in cryptography (bit commitment and NIZKs for NP) and communication complexity (randomized protocols for equality) in \Cref{sec:prelim-compiler} and we then present our protocol, prove its correctness, and discuss its extensions in \Cref{sec:protocol-compiler}.

\subsection{Preliminaries on Cryptography and Communication Complexity}
\label{sec:prelim-compiler}
\subsubsection{Random oracle and cryptographic objects}
\label{sec:ROM}
The random oracle model is the idealized version of hash functions in cryptography. In this model, we assume that all parties have a black-box access to a random function $H$. This model is fundamental in many cryptographic proofs. Moreover, in many of its uses, this random oracle can be replaced by a suitable cryptographic assumption, and one can prove the security of scheme against bounded adversaries. In other cases, such a formal statement is not known to hold, but still there are {\em heuristic implementations} of the random oracle. 

In our presentation, we will make use of the random oracle for building commitment schemes and (centralized) non-interactive zero-knowledge for NP. We discuss in \Cref{sec:discussion-compiler} how to avoid the Random Oracle, and its consequences to our protocol.

We describe now the needed cryptographic primitives.

\paragraph{Bit commitment} A commitment scheme is a two-phase  protocol between two parties, a sender and
receiver. In the first phase, the Sender wants to commit to a message $m$ to the Receiver, without revealing it. For that, the Sender sends the {\em commitment } $c$ to the Receiver. In the second
phase, the Sender will reveal the committed value $m$ by opening the commitment. We require two properties from the commitment scheme:
\begin{itemize}
    \item \textbf{hiding: } the Receiver cannot find $m$ from $c$
    \item \textbf{binding: } for every commitment $c$, there is a single message $m$ that can be opened.
\end{itemize}

It is well-known that there are no commitment schemes with unconditionally hiding
\emph{and} unconditionally binding properties in the plain model (i.e. without trusted assumptions), but we can construct them in the random oracle model. In this case, to commit to a message $m$, the Sender sends $H(m,r)$, for a random string $r$. It is folklore that such a protocol is perfectly binding and hiding.\footnote{The proof can be found, for example, in ~\cite{DodisNotes}.}

\paragraph{NIZKs} The so-called sigma-protocols are zero-knowledge proofs with the following structure:
\begin{enumerate}
	\item The prover sends a message $a$ to the verifier;
	\item The verifier sends a random challenge $c$ to the prover (and keeps no private randomness); and
	\item The prover answers back $z$ to the verifier, who then checks if the test passes or not
\end{enumerate}

This particular syntax is desired not only because of its simplicity, but there are constructions that allow us to directly compile such type of protocols into protocols with different structures. The most important of these compilers is the {\em Fiat-Shamir} transformation~\cite{FiatS86}. The goal of this transformation is to achieve a {\em non-interactive protocol} in the Random Oracle Model. The Fiat-Shamir transformation works roughly as follows.
\begin{enumerate}
	\item The prover computes the first message $a$ of the original sigma-protocol.
	\item The prover computes the third message of the original sigma-protocol, substituting the message from the verifier by  $H(a)$.
	\item The prover sends $(a,z)$ to the verifier, who then checks if the verifier of the sigma protocol would accept the transcript $a,H(a),z$.
\end{enumerate}

It is not hard to see that if the original sigma protocol is complete and zero-knowledge, so is the non-interactive one. Finally, it can be proved that if the challenge space has enough entropy, then the non-interactive protocol is also sound.

\subsubsection{Equality sub-protocol}

\begin{protocol}[Equality sub-protocol]\label{prot:equality} ~\\
	\noindent
	\textbf{Setup:} Each node $u$ has a value $m_u \in \{0,1\}^k$
	\begin{enumerate}
		\item The parties agree on an error-correcting code $\mathcal{C} : \{0,1\}^k \mapsto \{0,1\}^n$ with linear rate and distance
        \item Using shared randomness $i_1,...,i_t$, node $u$ sends $\mathcal{C}(m_u)_{i_1},\cdots,\mathcal{C}(m_u)_{i_t}$ to their neighbors 
        \item Each party compares its values with the ones sent by its neighbors and reject if they are different.
        \item Otherwise, the parties accept.
	\end{enumerate}
\end{protocol}

\begin{lemma}
    If $m_u = m_v$ for all $u,v \in V$, then the protocol accepts with probability $1$. If there exists a pair of neighbors $u,v$ such that $m_u \ne m_v$, then the protocol accepts with probability at most $2^{O(t)}$.
\end{lemma}

The proof of the lemma easily follows from the randomized protocol for equality studied in communication complexity~\cite{KushilevitzN97}.

\subsection{The Universal $\dzma$ Protocol}
\label{sec:protocol-compiler}

The main idea of our protocol is that the prover commits to the whole graph to
each of the nodes, and then proves to each node $u$ in zero-knowledge that $i)$ the committed
graph is consistent with the local view of $u$ and $ii)$ that the committed graph
has the property $P$. 
Finally, the nodes run a sub-protocol to verify the commitment of the graph sent by the prover is equal among neighbors.

Throughout this section, we assume that all the nodes agree on a canonical ordering $v_0,...,v_n$.  

\begin{protocol}[Universal $\dzma$ protocol] \label{prot:universal-zk} 
\noindent
  \begin{enumerate}
    \item To each node $u$, 
    the prover sends $\left(c^u_{i,j}\right)_{i,j\in [n]}$, where $c^u_{i,j}
      = c_{i,j} = comm(\mathbf{1}_{\{v_i,v_j\} \in E},r_{i,j})$. The prover also sends $v_i$ the opening of $c^{v_i}_{i,j}$ and $c^{v_i}_{j,i}$, for every $j \in [n]$

  \item The prover proves to each node $v_k$ in zero-knowledge that
    $\left(c^{v_k}_{i,j}\right)_{i,j\in [n]} \in L_P$, where 
    \[L_P = \{z \mid \exists (G' = (V,E'), r_{i,j}) , \text{ s.t. } z = comm(\mathbf{1}_{\{v_i,v_j\} \in E'},r_{i,j})_{i,j \in [n]} \wedge
       G' \text{ has property } P 
      \}\]
     
  \item All nodes perform the equality sub-protocol (\Cref{prot:equality}) with $t = \log n$ on
    $\left(c^u_{i,j}\right)_{i,j\in [n]}$, and reject if the sub-protocol rejects.
  \item Each node accepts iff the commitments of its edges open to the correct values, and both the zero-knowledge proof and the equality sub-protocol pass\label{step:verif2-universal}
\end{enumerate}
\end{protocol}

\thmuniversal

\subsection{Analysis and Zero-Knowledge of the Universal $\dzma$ Protocol}
We split the proof of \Cref{thm:universal-zk} in two parts. In \Cref{lem:universal-completeness} we show that \Cref{prot:universal-zk} is complete and sound, and analyze its complexity. Then, we prove the zero-knowledge property, in \Cref{lem:universal-zk}.

\begin{lemma}\label{lem:universal-completeness}
  \Cref{prot:universal-zk} has perfect completeness and soundness $O\left(\frac{1}{n}\right)$. Moreover, it follows the complexity indicated in \Cref{thm:universal-zk}.
\end{lemma}
\begin{proof}
  If the graph satisfies $P$ and the prover is honest, then we
  have that $a)$ the committed graph corresponds to the one in the network, $b)$
  each zero-knowledge proof is accepted, and $c)$ since the proof used the same
  commitment to each node, the equality sub-protocol also leads to acceptance.

  We now prove soundness. If the prover sends different $c^u_{i,j}$ and
  $c^v_{i,j}$ for {\em some} pair of vertices $u$ and $v$, then it these values must be different for some pair of two {\em neighbors} $u'$
  and $v'$, since the graph is connected. 
  In this case, the probability that $u'$ and $v'$ do not reject in the equality protocol is at most $O(\frac{1}{n})$.
  We continue now assuming that for all  $u \in V$, $i,j \in [n]$, there exists some $\hat{c}_{i,j}$ such that $c^u_{i,j} = \hat{c}_{i,j}$.

 Since the commitment scheme is perfectly binding, we have that if the prover is
  able to open the commitments with a view that is consistent with the local
  view of all nodes, then the committed graph indeed corresponds to the
  underlying network. In this case, we have that the committed graph does not
  have the property $P$, and therefore $\left(c^u_{i,j}\right)_{i,j\in [n]} \not\in
  L_P$. By the soundness property of the
  zero-knowledge protocol, each node accepts with probability at
  most $\negl(n)$.

Regarding the number of rounds of the protocol, we notice that the commitments can be sent at the same time as the first message of the centralized NIZK and therefore it does not incurs in any extra round of communication with the prover. Moreover, we notice that the required amount of communication from the prover is $O(n^2 \polylog(n))$ for the commitments and $\poly n$ for the NIZK proof. We notice that we cannot precise the precise polynomial, since this depends on the property being proven. Finally, the communication between the nodes is $\log n$.
\end{proof}

We now prove the zero-knowledge property. 
We notice that we consider zero-knowledge against any set $M$ of malicious nodes.
In this setting, the simulator has to create a transcript for all the nodes of $M$ together, and so it is given the neighborhood of all nodes in $M$. We define our simulator in \Cref{prot:sim-universal-zk}.

\begin{protocol}[Simulator for a set $M$ of malicious nodes]
\label{prot:sim-universal-zk}~ 
  \begin{enumerate}
    \item Compute $\left(\hat{c}_{i,j}\right)_{i,j \in [n]}$ where $\hat{c}_{i,j}
      = \begin{cases} 
        comm(1_{\{i,j\} \in E},r_{i,j}), & \text{if $i \in M$ or
        $j \in M$}\\
      comm(0,r_{i,j}), & \text{otherwise} \end{cases}
                      $

    \item Send $\left(\hat{c}_{i,j}\right)_{i,j \in [n]}$ to
      every $u \in M$, along with the opening of  $\hat{c}_{u,j}$ and $\hat{c}_{j,u}$, for every $j$
    \item Run the zero-knowledge simulation with every node in $M$
    \item Run the simulation for the equality protocol
    \item Output the transcript of the simulation
  \end{enumerate}
\end{protocol}

\begin{lemma}\label{lem:universal-zk}
Let  $\mathcal{R}$ be the distribution of transcripts in the real run of the protocol and $\mathcal{S}$ be the distribution of transcripts from the simulator described in 
\Cref{prot:sim-universal-zk}. We have that for every polynomial-time algorithm $\mathcal{D}$:
\[\Pr_{\tau_R \sim \mathcal{R}}[\mathcal{D}(\tau_R) = 1] - \Pr_{\tau_S \sim \mathcal{S}}[\mathcal{D}(\tau_S) = 1] = \negl(\lambda).\]
\end{lemma}
\begin{proof}
 In order to prove this statement, we consider the following hybrids:
  \begin{description}
    \item[Hybrid 0:] The transcript of the real run of \Cref{prot:universal-zk}
    \item[Hybrid 1:] The same as hybrid $0$, but the equality protocol is replaced by its simulations using the provided commitment to every node $u$ in the network.
    \item[Hybrid 2:] The same as hybrid $1$, but instead of the prover and the nodes
      in $M$ running the zero-knowledge protocol, they run the simulation of
      the zero-knowledge protocol.
    \item[Hybrid 3:] The same as hybrid $2$, but the prover replaces the
      commitments of the graph by commitments of $\hat{c}_{i,j}$
    (\Cref{prot:sim-universal-zk}).
  \end{description}

  We show now that the output distribution of \textbf{Hybrid i} is computationally close to the output distribution of \textbf{Hybrid~i+1}. For that, let us define $\mathcal{H}_i$ as the distribution on the transcripts of \textbf{Hybrid i}. Notice that $\mathcal{R} = \mathcal{H}_0$ and $\mathcal{S} = \mathcal{H}_3$.

  \paragraph{Hybrid 0 vs. Hybrid 1.}  
  In an honest run of \Cref{prot:universal-zk}, the prover provides the same commitments to all nodes.
  Hence, the simulation of the equality protocol with the commitment provided to node $u$ has the exact same output distribution as the original protocol. Therefore, for any (possibly unbounded) $\mathcal{D}$ 
    \[\Pr_{\tau_0 \sim \mathcal{H}_0}[\mathcal{D}(\tau_0) = 1] - \Pr_{\tau_1 \sim \mathcal{H}_1}[\mathcal{D}(\tau_1) = 1] = 0.\]
    
  \paragraph{Hybrid 1 vs. Hybrid 2.}  
  We show now that distinguishing the output of \textbf{Hybrid 1} and \textbf{Hybrid 2} enables to break the underlying centralized zero-knowledge protocol.
  By assumption, let $\mathcal{D}$ be a polynomial-time algorithm and $p$ be a polynomial such that
  \[\left| \Pr_{\tau_1 \sim \mathcal{H}_1}[\mathcal{D}(\tau_1) = 1] - \Pr_{\tau_2 \sim \mathcal{H}_2}[\mathcal{D}(\tau_2) = 1]\right|   \geq \frac{1}{p(n)}.\]
  
  We can construct a distinguisher $\mathcal{D}'$ for the zero-knowledge protocol for $L_P$ as follows.
  \begin{enumerate}
    \item On input $\left(\hat{c}_{i,j}\right)_{i,j}$ and transcript $\pi$, run the simulation of the communication of all nodes in $M$ and compute $\tau_{EQ}$
    \item Run $\mathcal{D}$ on input $ \tau = (\left(\hat{c}_{i,j}\right)_{i,j},\pi,\tau_{EQ})$.
  \end{enumerate}
  
  Notice that $\mathcal{D}'$ also perfectly simulates the interaction with the neighbors and therefore the success probability of $\mathcal{D}'$ is exactly the same as the one of $\mathcal{D}$. Therefore, there exists a polynomial-time algorithm $\mathcal{D}'$ that distinguishes the real transcript and the simulator for the zero-knowledge protocol for $L_P$, which is a contradiction.

  \paragraph{Hybrid 2 vs. Hybrid 3.} 
  We will consider the sub-hybrids $2.i.j$, for $i,j \in [n+1]$.  In each sub-hybrid, we will change the commitments that are not opened to $M$ in the real protocol by commitments of $0$, one by one.

  \begin{description}
    \item[Hybrid $2.i^*.j^*:$] The commitments are computed as following:
    \[
\hat{c}_{i,j}
      = \begin{cases} 
        comm(1_{\{v_i,v_j\} \in E},r_{i,j}), & \text{if $v_i \in M$ or
        $v_j\in M$ or $i > i^*$ or ($i = i^*$ and $j > j^*$)}\\
      comm(0,r_{i,j}), & \text{otherwise} \end{cases}. \]
   Then, we continue with the simulation of the zero-knowledge proof and the simulation of the equality protocol as in Hybrid $2$.
  \end{description}

Notice that $a)$ in Hybrid $2.0.0$, all committed values correspond to the edges in the protocol, which is the exact same setting of Hybrid $2$; $b)$ Hybrid $2.k.n$ is exactly the same as Hybrid $2.k+1.0$; and $c)$ Hybrid $2.n.n$ is equivalent to Hybrid $3$, since the unopened commitments that are given to $M$ contain the value $0$. 

  Let $\mathcal{D}$ be an arbitrary polynomial-time algorithm and let us define
  \[
  \eps_{i,j} = \begin{cases}
  \left| \Pr_{\tau_{i,j} \sim \mathcal{H}_{2.j.k}}[\mathcal{D}(\tau_{i,j}) = 1] - \Pr_{\tau_{i,j+1} \sim \mathcal{H}_{2.j.k+1}}[\mathcal{D}(\tau_{i,j+1}) = 1]\right|,& \text{if } k < n \\
   \left|  \Pr_{\tau_{i,j} \sim \mathcal{H}_{2.j.k}}[\mathcal{D}(\tau_{i,j}) = 1] - \Pr_{\tau_{i,j+1} \sim \mathcal{H}_{2.j+1.0}}[\mathcal{D}(\tau_{j+1,1}) = 1]\right|,& \text{if } k = n
    \end{cases}
  \]  

  We have that by triangle inequality
    \[\left|\Pr_{\tau_3 \sim \mathcal{H}_3}[\mathcal{D}(\tau_3) = 1] - \Pr_{\tau_2 \sim \mathcal{H}_2}[\mathcal{D}(\tau_2) = 1] \right| \leq \sum_{i,j} \eps_{i,j}.\]

  We now show that if there is a graph $G$, a set $M$ of malicious nodes in $G$, two indices $i,j$ and a polynomial $p$ such that $\eps_{i,j} \geq \frac{1}{p(n)}$, then the commitment scheme is not computationally hiding against non-uniform adversaries. Therefore, assuming the hiding property of the commitment scheme, we can conclude that 
      \[\left|\Pr_{\tau_3 \sim \mathcal{H}_3}[\mathcal{D}(\tau_3) = 1] - \Pr_{\tau_2 \sim \mathcal{H}_2}[\mathcal{D}(\tau_2) = 1]\right| =  \negl(n).\]

  To finish the proof, let us suppose that $\eps_{i,j} \geq \frac{1}{p(n)}$, where $G,M,i,j,p$ are as above.
  We define the following non-uniform adversary $\mathcal{A}$ that breaks the hiding property of the commitment scheme, thus reaching a contradiction.
  \begin{enumerate}
      \item The values of $G,M,i,j$ are provided as non-uniform advice to $\mathcal{A}$
      \item $\mathcal{A}$ is given the challenge commitment $c$
      \item $\mathcal{A}$ computes 
          \[
\hat{c}_{i,j}
      = \begin{cases} 
        comm(1_{\{v_i,v_j\} \in E},r_{i,j}), & \text{if $i \in M$ or
        $j \in M$ or $i > i^*$ or ($i = i^*$ and $j > j^*$)}\\
                c, & \text{if $i = i^*$ and $j = j^*$}\\
      comm(0,r_{i,j}), & \text{otherwise} \end{cases} \]
      \item $\mathcal{A}$
        follows the simulation of the zero-knowledge proof and the simulation of the equality protocol as in Hybrid $2$
    \item $\mathcal{A}$ runs $\mathcal{D}$ on the computed values
  \end{enumerate}

We have that 
\begin{align*}
&\left|\Pr[\mathcal{A}(comm(0,r')) = 1] - \Pr[\mathcal{A}(comm(1,r')) = 1]\right|
\\
 &=  \left| \Pr_{\tau_{i,j} \sim \mathcal{H}_{1.j.k}}[\mathcal{D}(\tau_{i,j}) 
=1] - \Pr_{\tau_{i,j+1} \sim \mathcal{H}_{1.j.k+1}}[\mathcal{D}(\tau_{i,j+1}) = 1]\right| = \eps_{i,j},
\end{align*}
where the first equality follows from construction and the second equality by contradiction. Since we assume that our protocol is hiding, we have that $\eps_{i,j} = \negl(n)$.
\end{proof}

\subsection{Discussion and Extensions of the Universal $\dzma$ Protocol}
\label{sec:discussion-compiler}

We notice that as in our other protocols, we can tune the parameter $t$ as a tradeoff between communication between neighbors and soundness. 
For instance,  increasing the communication between neighbors to poly-logarithmic allows to achieve negligible soundness error. Alternatively, if constant soundness is sufficient then each node only needs to send $O(1)$ bits to each of its neighbors.

One downside of our protocol is that it works in the random oracle model. We notice that we can build NIZKs and bit-commitments under standard cryptographic assumptions~\cite{Naor89,BlumFM88,PeikertS19}, and therefore we can also achieve $\dzma$ for NP which is ZK against \emph{computationally bounded} malicious nodes {\em without using a random oracle}. 

As we mentioned in \Cref{sec:intro-compiler}, \cite{BickKO22} showed a compiler from PLS to distributed zero-knowledge proofs. 
We notice that this compiler can also be instantiated with the trivial PLS for NP graph properties where the prover sends the whole graph along with the NP proof, and the nodes verify if the graph is correct and locally check if it has the desired property.  
Their protocol builds upon a PLS of $\ell$-bit certificates, which is verified by each node using a circuit of size $s$.
Our protocol is based on a (centralized) NIZK $\pi$ that proves that the committed graph has the desired property.
\Cref{tab:comparison} compares the parameters of the two compilers. 
Roughly, their compiler is much more efficient for properties that admit small PLS with fast local verification and requires no cryptographic assumptions,
while our protocol achieves strong zero-knowledge against any number of colluding parties, and guarantees low-communication among neighbors regardless of the problem in hand.

\begin{table}[]
    \centering
    \begin{tabularx}{\textwidth}{|X|X|X|}
     \hline
         & \cite{BickKO22} & Our compiler \\ \hline
      Assumption   & \textcolor{blue}{None} & \textcolor{red}{ROM or Cryptographic assumption} \\ \hline
      Type of distributed zero-knowledge  & \textcolor{red}{Weak} & \textcolor{blue}{Strong} \\ \hline
      Max. size of coalition  ($k$)  & \textcolor{red}{ $k = o(n)$} & \textcolor{blue}{$n$} \\ \hline
      Communication with prover   & 
      \textcolor{teal}{$1$ round \newline $O(k^2 (\Delta \ell + s)\log n)$-bits}
       & \textcolor{teal}{1 round \newline $O(n^2 + |\pi|$)-bits} \\ \hline
      Communication with neighbors   & \textcolor{red}{$O(k)$ rounds} \newline \textcolor{red}{$O(k^2 (\Delta \ell + s)\log n)$-bits}  & \textcolor{blue}{ 1 round}  \newline \textcolor{blue}{$O(\log n)$-bits} \\ \hline
    \end{tabularx}
    \caption{Comparison between the compilers of \cite{BickKO22} and the one from this work.
    We highlight the \textcolor{blue}{advantages}, \textcolor{red}{disadvantages}, and \textcolor{teal}{equivalence/incomparability}.
    For the type of distributed zero-knowledge we use the notation from~\cite{BickKO22}; our work only considers their {\em strong} version of distributed zero-knowledge.} 
    \label{tab:comparison}
\end{table}

We would like to end this comparison by stressing that in both protocols, the communication between prover and nodes is polynomial, but the precise polynomial depends on the parameters from the base protocol (PLS in their case, NIZK in ours).

Finally, we also notice that our compiler can be easily extended to achieve, under computational assumptions,  distributed {\em interactive} zero-knowledge proofs for any graph property that can be verified in PSPACE. 
This follows from the fact that under the assumption that bit-commitments exist, the set of problems that have interactive computational zero-knowledge proofs is equal to PSPACE~\cite{BenOrGGHKMR88}. As in the $\dzma$ case, the round and communication complexity of the protocol will intrinsically depend on the graph property being proven.

\begin{corollary}
   Any graph property in PSPACE admits a distributed computational zero-knowledge proof.
\end{corollary}

\bibliographystyle{alpha}
\bibliography{zkbib}

\end{document}